\documentclass[superscriptaddress,aps,pra,twocolumn,showpacs,nofootinbib,longbibliography]{revtex4-1}
\usepackage{etex}
\usepackage{amsmath,amssymb,amsthm}
\usepackage[colorlinks=true,citecolor=blue,urlcolor=blue]{hyperref}
\usepackage[pdftex]{graphicx}
\usepackage{times,txfonts}
\usepackage{soul}
\usepackage{braket}
\usepackage{color}
\usepackage{graphics,epstopdf}
\usepackage{natbib}
\usepackage{slashbox}

\newcommand{\be}{\begin{equation}}
\newcommand{\ee}{\end{equation}}
\newcommand{\ba}{\begin{eqnarray}}
\newcommand{\ea}{\end{eqnarray}}
\newcommand{\ketbra}[2]{|#1\rangle \langle #2|}
\newcommand{\tr}{\operatorname{Tr}}

\newtheorem{definition}{Definition}
\newtheorem{proposition}{Proposition}

\begin{document}

\title{Superunsteerability as a quantifiable resource for random access codes assisted by Bell-diagonal states}  

\author{C. Jebarathinam}
  \email{jebarathinam@gmail.com}
 \affiliation{S. N. Bose National Centre for Basic Sciences, Block JD, Sector III, Salt Lake, Kolkata 700 106, India}   
 \affiliation{Department of Physics and Center for Quantum Frontiers of Research \& Technology (QFort), National Cheng Kung University, Tainan 701, Taiwan}
 \author{Debarshi Das}
\email{dasdebarshi90@gmail.com}
\affiliation{Centre for Astroparticle Physics and Space Science (CAPSS), Bose Institute, Block EN, Sector V, Salt Lake, Kolkata 700 091, India}

 \author{Som Kanjilal}
\affiliation{Centre for Astroparticle Physics and Space Science (CAPSS), Bose Institute, Block EN, Sector V, Salt Lake, Kolkata 700 091, India}

\author{R.   Srikanth} 
\address{Poornaprajna Institute of  Scientific Research Bangalore-
  560    080,   Karnataka,    India}

\author{Debasis  Sarkar}
\author{Indrani Chattopadhyay}
\affiliation{Department of Applied Mathematics, University of Calcutta, 92, A.P.C. Road, Kolkata-700009, India.}

   \author{A.    S.     Majumdar}
\email{archan@bose.res.in}
   \affiliation{S. N. Bose National Centre for Basic Sciences, Block JD, Sector III, Salt Lake, Kolkata 700 106, India} 
   
\begin{abstract}
We show how nonclassical correlations in local bipartite states can act as a resource for quantum information processing. Considering the task of quantum random access codes (RAC) through separable Bell-diagonal states, we demonstrate the advantage of superunsteerability over classical protocols assisted with two-bits  of shared randomness. We propose a measure of superunsteerability, which 
quantifies nonclassicality beyond quantum steering, and obtain its analytical expression for Bell-diagonal states in the context of the two- and three-setting steering scenarios that are directly related to the quantum $2 \to 1$ and $3 \to 1$ RAC protocols, respectively. The maximal values of our quantifier yield the optimal quantum efficiency for both of the above protocols, thus showing that superunsteerability provides a precise characterization of the nonclassical resource for implementing RACs with separable Bell-diagonal class of states. 
\end{abstract}

\maketitle 
   \date{\today}
   
   \section{Introduction}  It is well-known that quantum entanglement  in  states of composite systems leads to stronger than classical correlations, such as Bell nonlocality \cite{bel64} and  quantum steering \cite{Sch35,WJD07}. Nonclassical features of entanglement provide quantum advantage for information processing tasks.  Bell nonlocality certifies the presence of entanglement in a device-independent way, serving as resource for device-independent quantum key distribution \cite{BHK05}, and generation of genuine randomness
 \cite{PAM+10}. Quantum steering which is a weaker form of nonclassical correlations  \cite{WJD07}, is a useful resource for one-sided device-independent quantum key distribution \cite{BCW+12} and generation of  randomness \cite{PCS+15}. The violation of a steering inequality  certifies  entanglement in the scenario where measurement devices on only one  side are trusted. 

Though entanglement is used to accomplish a multitude of quantum information tasks, it is an 
expensive resource, making it hard to prepare and preserve the required nonlocal correlations. 
Nonclassicality, on the other hand, is not restricted to entangled states. Certain separable states are known to provide quantum advantage for  tasks such as deterministic quantum computation \cite{DSC08} and quantum cryptography \cite{Pir14}. It has been argued that quantum discord \cite{OZ01,HV01} characterizes such nonclassicality of separable states \cite{GHA12,RP14}.  Nonvanishing quantum discord enables better efficiency over classical protocols using limited shared randomness for generating random access codes (RACs) using  bipartite separable states \cite{BP14}, though the amount of discord \cite{DVB10} does not specify the usefulness of the corresponding state. Additionally, it is possible to extract operationally nonlocal features of separable states through measurement \cite{LF11}.

In realistic information processing scenarios shared classical randomness may not be a free resource. 
When the amount of accessible randomness is limited, separable states may provide further benefit over classical protocols \cite{DW15,CJ,JAS17,Jeb17,JDS+18,DBD+17,DAS201855}. In such scenarios the notion of superlocality \cite{DW15,JAS17,JDS+18} refers
to the requirement for a larger dimension of the preshared randomness to simulate the respective local correlations, than that of the quantum systems generating them. Recently, by generalizing the concept of superlocality, the quantumness of certain unsteerable correlations has been pointed out by defining the notion of superunsteerability \cite{DBD+17,DAS201855}. It is the requirement for a larger dimension of the classical variable that the steering party has to preshare with the trusted party for simulating the unsteerable correlations than that of the quantum systems reproducing them. It is thus pertinent to ask as to what extent the domain of quantum advantage may further be extended by enabling practical tasks by even unsteerable states. 

In the present work we show that it is indeed possible to perform useful
quantum information processing tasks using the  nonclassicality associated with superunsteerability of certain class of separable states as resource in situations where an unlimited or infinite amount of preshared randomness is not available classically. Specifically, we consider the task of generating RACs which in addition to being a primitive information processing task, are of importance in understanding the scope of physical theories \cite{PS15}. RACs have been used to demarcate the boundary between quantum and post-quantum theory in terms of the information causality principle \cite{RFBW14}. The encoding and decoding measurement strategies employed in the RACs using shared bipartite states \cite{BP14}, resemble closely the scenario of quantum steering. This enables us to recast the $2 \to 1$ and $3 \to 1$ RAC protocols in terms of the 2- and 3-setting steering scenarios, respectively. 

Here we focus on the above RAC protocols assisted by a  generic class of mixed states, {viz.} Bell-diagonal states which have been widely used in information processing \cite{PM19}. By formulating a measure of
superunsteerability as an information theoretic resource,  we show that in the presence of two-bits of shared classical randomness,  superunsteerability identifies the optimal separable states within Bell-diagonal class of states for the quantum RACs.  Our analysis addresses an important  question as to which is the precise resource responsible for enabling RACs assisted by Bell-diagonal states without access to entanglement and infinite shared randomness.

The rest of the paper is arranged as follows. In Section \ref{sec2} we provide brief discussions on quantum steering, superunsteerability and  RAC protocols. In Section \ref{sec3} we propose a quantifier of superunsteerability, termed as ``Schr\"{o}dinger strength", and obtain its analytical expressions  for Bell-diagonal states in the context of the two- and three-setting steering scenarios. How superunsteerability serves as resources for quantum RACs using separable Bell-diagonal states is presented in Section \ref{sec4}. Finally, we conclude in Section \ref{sec5}.

   \section{Quantum steering and random access codes}\label{sec2}
   
 Consider a  scenario where Alice and Bob share an unknown quantum system
described by $\rho_{AB}\in \mathcal{B}(\mathcal{H}_A \otimes \mathcal{H}_B)$, with Alice 
performing a set of black-box measurements which produces a set of conditional states on Bob's side. 
Such a scenario is called one-sided device-independent since
Alice's measurement operators ${\bf{M}}_A:=\{M_{a|x}\}_{a,x}$  are unknown. The steering scenario 
is completely characterized by an assemblage \cite{Pus13} ${\pmb{\sigma}}:=\{\sigma_{a|x}\}_{a,x}$
which is the set of unnormalized conditional states on Bob's
side. Each element
in the assemblage ${\pmb{\sigma}}$ 
is given by $\sigma_{a|x}=p(a|x)\rho_{a|x}$,  where $p(a|x)$ is the conditional probability of getting the outcome of Alice's measurement  and
$\rho_{a|x}$ is the normalized conditional state on Bob's side.
Quantum theory predicts the assemblage as $\sigma_{a|x}=\tr_A ( M_{a|x} \otimes \openone \rho_{AB})$, $\forall$ $\sigma_{a|x} \in {\pmb{\sigma}}$.
Suppose Bob performs a set of known positive operator valued measurements (POVM) ${\bf{M}}_B:=\{M_{b|y}\}_{b,y}$ on ${\pmb{\sigma}}$. Then the
scenario is characterized by the box (or, correlation) $P(ab|xy)$= $\{p(ab|xy)\}_{a,x,b,y}$  = $\Big\{ \tr\big[M_{b|y} \sigma_{a|x} \big] \Big\}_{a,x,b,y}$, which is a set of joint probability distributions $p(ab|xy)$ for all possible $a$, $x$, $b$ and $y$.

A box $P(ab|xy)$ detecting steerability from Alice to Bob
 does not have a decomposition of the form given by 
\be
P(ab|xy)= \sum^{d_\lambda-1}_{\lambda=0} p(\lambda) P(a|x,\lambda) P(b|y; \rho_\lambda). \label{LHV-LHS}
\ee
Here the box $P(a|x,\lambda)$ = $\{p(a|x,\lambda)\}_{a,x}$ is the set of arbitrary probability distributions $p(a|x,\lambda)$ conditioned upon shared randomness/hidden variable $\lambda$ occurring with the probability $p(\lambda)$; $\sum^{d_\lambda-1}_{\lambda=0} p(\lambda) =1$. On the other hand, $P(b|y; \rho_\lambda)$ = $\{p(b|y; \rho_\lambda) \}_{b,y}$ is the set of quantum probability distributions $p(b|y; \rho_\lambda)=\tr(\Pi_{b|y}\rho_\lambda)$, arising from some local hidden state $\rho_\lambda$, and
$d_\lambda$ denotes the dimension of the shared randomness/hidden variable $\lambda$.
The above decomposition is called a local hidden variable-local hidden state (LHV-LHS) model.
If the correlation arising from the given steering scenario does not have steerability, it can still have nonclassical correlation 
when there is a restriction on the amount of shared randomness \cite{DBD+17,DAS201855}.
 Superunsteerability  \cite{DBD+17,DAS201855} is defined as the requirement for a larger dimension of the classical variable that the steering party (Alice) preshares with the trusted party (Bob) for simulating the given unsteerable correlations, than that of the quantum state which reproduces them.
Suppose we have a quantum state in $\mathbb{C}^{d_A}\otimes\mathbb{C}^{d_B}$
and measurements which produce a unsteerable bipartite  box $P(ab|xy)$ := $\{ p(ab|xy) \}_{a,x,b,y}$.
Then, superunsteerability (SU) holds if and only if (iff) there is no decomposition of the box in the form given by Eq. (\ref{LHV-LHS}),
with  $d_\lambda\le d_A$.  Superunsteerability provides an operational characterization to the quantumness of unsteerable  boxes \cite{DBD+17}. Note that superunsteerability is defined in the standard steering scenario. The only difference is that there is a constraint on the dimension of the resources producing the correlations in the context of superunsteerability. In the standard steering scenario, the dimension of the hidden variable and the Hilbert space 
dimension of the steering party's side are unbounded which implies that any unsteerable correlation 
can be reproduced by a LHV-LHS model as well as by performing local measurements on an appropriate  separable quantum state which admits a
classical-quantum decomposition \cite{CS17}. 
In the superunsteerability scenario (where the dimension of the resource is restricted),
if the given unsteerable correlation cannot be simulated by a LHV-LHS model, but can be simulated by 
local measurements on a quantum system, then the given unsteerable correlation is superunsteerable. Hence, superunsteerability refers to dimensional advantage in simulating certain unsteerable correlations using quantum resources.

Random access codes (RACs) are a class of communication tasks where one party tries to guess the data held by the other party with limited amount of communication. A quantum version of RACs was constructed and shown to provide advantage over  classical RACs  \cite{AAA+02}. Here we focus on RACs where
a shared bipartite quantum state is used as resource \cite{PZ10}. Recently, it has been shown \cite{BP14} that in the  case when shared classical randomness is not a free resource, even certain separable states  may provide quantum advantage. 
Protocols for quantum $3 \to 1$ and $2 \to 1$ RACs  assisted with separable  two-qubit Bell-diagonal states  provide better efficiencies compared to classical protocols aided with two-bits of shared  randomness. However, for these two RACs, the optimal quantum advantage is provided by different  Bell-diagonal states possessing inequivalent values of geometric discord \cite{DVB10}. 

In an RAC denoted by $n\xrightarrow{p}1 $, Alice has $n$ bits ($n=2,3$) of information, 
$x=(x_1\cdots x_n)$, which she encodes in a $1$ bit communication. Alice  shares a pair of correlated qubits with Bob and encodes her $n$ bits into her part of the shared qubit by performing a quantum measurement. She communicates $1$ classical bit which is the result of this quantum measurement to Bob who  figures out the $i$th bit, $x_i$, of
Alice's input string $x$ with  probability ${\rm Pr} ({{b}_{i}}={{x}_{i}})$ by performing a suitable quantum measurement.  
 The efficiency of the RAC is given by the probability ${{P}_{{\rm min} }}$ of Bob's correct guess in the worst-case scenario, i.e., ${{P}_{{\rm min} }}={{{\rm min} }_{x,i}}{\rm Pr} ({{b}_{i}}={{x}_{i}})$.  If no randomness is allowed in this scenario, ${{P}_{{\rm min} }}=0$, as there is always a bit that Bob guesses wrongly. In the presence of shared classical randomness $r$ which occurs with probability $p_r$, the efficiency ${{P}_{{\rm min} }}$ is additionally averaged over the assisting random bits, ${{P}_{{\rm min} }}={{{\rm min} }_{x,i}}{{\sum }_{r}}{{p}_{r}}{\rm Pr} ({{b}_{i}}={{x}_{i}}|r)$.  A classical $2 \xrightarrow{p}1$ RAC assisted with two bits  has  ${{P}_{{\rm min} }}\leq \frac{2}{3}$, and  for the similar
 $3 \xrightarrow{p}1$ RAC,  ${{P}_{{\rm min} }}\leq \frac{1}{2}$ \cite{BP14}.

Note that  in the above two RACs, information about inputs of Alice are contained in the conditional states of Bob who tries to decode this information  by performing measurements on these  states. Such a scenario of decoding and encoding measurement strategies  of Alice and Bob in the RACs resembles that of the steering scenario. The  encoding operation of $2 \to 1$ RAC makes use of correlations present in the two-qubit states along the $x$ and $y$ axes, and for $3 \to 1$ RAC, the  encoding operation makes use of correlations present in the two-qubit states along the $x$, $y$ and $z$ axes. It is thus natural to ask whether quantum correlation as quantified by superunsteerability in the two- and three-setting steering scenarios can identify optimal separable Bell diagonal  states for the $2 \to 1$ and $3 \to 1$ RACs, respectively.  The above question motivates us to define a measure of superunsteerability for the given steering scenarios.

\section{Quantifying nonclassicality beyond steering}\label{sec3} 

A category of nonclassicality of local correlations captured by superlocality \cite{DW15} has been quantified by a measure called Bell strength  \cite{JAS17}. Analogously,  we define a quantity called Schr\"{o}dinger strength  to quantify the nonclassicality of unsteerable boxes which demonstrate superunsteerability. Note that a box $P(ab|xy)$ in the $n$-setting steering scenario is the set of joint probabilities $p(ab|xy)$ for all possible $a$, $b$ and for all $x$ $\in$ $\{0, 1, 2, ..., n-1\}$ and $y$ $\in$ $\{0, 1, 2, ..., n-1\}$. In this scenario any box $P(ab|xy)$ can be decomposed into a convex mixture of a steerable part and an unsteerable part,
\be
P(ab|xy)=p P_{S}(ab|xy)+(1-p) P_{US}(ab|xy), \label{steerstrength}
\ee
where $P_{S}(ab|xy)$ is a steerable box and $P_{US}(ab|xy)$  is an unsteerable box which may be superunsteerable; $0 \leq p \leq 1$. 
The weight of the box $P_{S}(ab|xy)$ minimized over all possible decompositions of 
the  form (\ref{steerstrength}) is called steering cost of the box $P(ab|xy)$ furnishing a measure of quantum steering \cite{DDJ+18}.

On the other hand, the weight of the box $P_{S}(ab|xy)$ maximized over all possible decompositions of 
the  form (\ref{steerstrength}) is called the Schr\"{o}dinger strength of the box $P(ab|xy)$, given by
\be
S_{SU_n} \Big( P(ab|xy) \Big):=\max_{\text{decompositions}}p.
\ee
Here, $0 \leq S_{SU_n} \leq 1$ (since, $0 \leq p \leq 1$). 
The optimal decomposition that gives the Schr\"{o}dinger strength of the box is called the canonical decomposition in which the steerable part is an extremal steerable box. Such optimal decomposition is analogous to the optimal decomposition that gives Svetlichny strength of a tripartite correlation studied in Refs. \cite{Jeb17,JDS+18} which quantifies 
genuine nonclassicality.

To illustrate the above minimization and maximization to obtain the steering cost and  Schr\"{o}dinger strength, respectively, of a  box $P(ab|xy)$, let us consider the following
family of correlations:
\be \label{BB84}
P(ab|xy)=\frac{1+(-1)^{a+b+xy}\delta_{x,y}V}{4}, \quad a,b,x,y \in \{0,1\}
\ee
which is called white-noise BB84 family in Ref. \cite{DDJ+18}. 
In the context of a steering scenario where Alice performs two black-box dichotomic 
measurements and Bob performs two projective qubit measurements in mutually unbiased
bases, the above box 
detects quantum steering for $V > 1/\sqrt{2}$ and superunsteerability for $V>0$
as shown in Ref. \cite{DBD+17}. In Ref. \cite{DDJ+18}, it has been demonstrated
that the optimal decomposition that gives the steering cost of the box (\ref{BB84})
is given by 
\be \label{BB84dec}
P(ab|xy)=\frac{\sqrt{2}V-1}{\sqrt{2}-1}P_{Ext}(ab|xy)+ \left(1- \frac{\sqrt{V}-1}{\sqrt{2}-1}\right) P_{US}(ab|xy),
\ee
with $V \ge 1/\sqrt{2}$, 
where 
\be
P_{Ext}(ab|xy)=\frac{1+(-1)^{a+b+xy}\delta_{x,y}}{4}, \quad a,b,x,y \in \{0,1\}
\ee
is an extremal steerable box and $P_{US}$ is an unsteerable box which is the box (\ref{BB84}) 
with $V=1/\sqrt{2}$.

Even for $ 0 < V \le 1/\sqrt{2}$, the box (\ref{BB84}) can be decomposed as a convex mixture 
of the extremal steerable box and an unsteerable box as follows:
\be \label{BB84dec1}
P(ab|xy)=V P_{Ext}(ab|xy)+ (1-V) P_N(ab|xy),
\ee
where $P_N(ab|xy)$ is the white noise for which $P_N(ab|xy)=1/4$ for all $a,b,x,y$.
Note that the unsteerable box in the decomposition given by Eq. (\ref{BB84dec})
still can be decomposed as a convex mixture of the extremal steerable box and the white noise
to obtain the decomposition given by Eq. (\ref{BB84dec1}).
Moreover, the white noise can only be decomposed in terms of extremal 
steerable boxes as 
an uniform mixture of suitable extremal steerable boxes.
Thus, the weight of the extremal steerable box in the decomposition given by Eq. (\ref{BB84dec1})
is maximal over all possible decomposition. Such a maximized weight of the extremal steerable box
is called the Schr\"{o}dinger strength.
Note that even when the box (\ref{BB84}) has the steering cost equals to zero for $0 < V \le 1/\sqrt{2}$, 
it has a nonvanishing Schr\"{o}dinger strength  in this range  which implies the presence of superunsteerability.

Before proceeding, we want to mention that in the $n$-setting steering scenario, Alice's measurement settings and Bob measurement settings are denoted by $A_x$ and $B_y$ respectively, where $A_x$ $\in$ $\tilde{A}$ = $\{A_0, A_1, A_2, ..., A_{n-1} \}$ and $B_y$ $\in$ $\tilde{B}$ = $\{ B_0, B_1, B_2, ..., B_{n-1} \}$. The sets of Alice's and Bob's measurement settings are denoted by $\tilde{A}$ and $\tilde{B}$ respectively. We will now present the definition of Schr\"{o}dinger strength of a bipartite state.

\begin{definition}
 The Schr\"{o}dinger strength of a bipartite state $\rho$ is defined as:
 \be
 S_{SU_n}(\rho)=\max_{\tilde{A}, \tilde{B}} S_{SU_n}\Big(P(ab|xy;\rho) \Big).
 \ee
Here $S_{SU_n}\Big(P(ab|xy;\rho) \Big)$ is the Schr\"{o}dinger strength of the box $P(ab|xy;\rho)$ = $\{ p(ab|xy;\rho) \}_{a,x,b,y}$, where $p(ab|xy;\rho)$ = Tr$[ (M_{a|x} \otimes M_{b|y}) \rho ]$ is the joint probability distribution of getting the outcomes $a$ and $b$ when measurements $A_x$ and $B_y$ are performed locally by Alice and Bob, respectively, on the shared bipartite state $\rho$. $M_{a|x}$ is the measurement operator corresponding to the measurement settings $A_x$ and outcome $a$. $M_{b|y}$ is defined similarly. Here the maximization is taken over all possible sets $\tilde{A}$ = $\{A_0, A_1, A_2, ..., A_{n-1} \}$ and $\tilde{B}$ = $\{ B_0, B_1, B_2, ..., B_{n-1} \}$.
\end{definition}

Let us now evaluate the Schr\"{o}dinger strength of 
the Bell-diagonal two-qubit states in the following two scenarios:
(i) Alice performs two black-box dichotomic measurements and Bob performs 
projective qubit measurements corresponding to two arbitrary mutually unbiased bases
and 
(ii) Alice performs three black-box dichotomic measurements and Bob performs 
projective qubit measurements corresponding to three arbitrary mutually unbiased bases. 
For these two scenarios, the  
following linear steering inequalities \cite{CJW+09}:
\be
\frac{1}{\sqrt{n}} \left|\sum^n_{k=1} \braket{A_k \otimes B_k}\right| \le 1, \label{FJWR}
\ee
with $n=2,3$, respectively, can be used to witness steerability of two-qubit states from Alice to Bob.
Here, $\braket{A_k \otimes B_k}=\tr \left( A_k \otimes B_k \rho_{AB}\right)$.

The Bell-diagonal states admit the form,
\be
\tau=\frac{1}{4}\left(\openone \otimes \openone +\sum_i c_i \sigma_i \otimes \sigma_i \right)=\sum_{ab} \lambda_{ab} \ketbra{\beta_{ab}}{\beta_{ab}}, \label{BD}
\ee
where $\lambda_{ab}$, here $a,b\in \{0,1\}$, denote the eigenvalues of the Bell-diagonal states which are given by 
$\lambda_{ab}=\frac{1}{4}\left[1+(-1)^a c_1 - (-1)^{a+b} c_2+(-1)^b c_3 \right]$
and $\ket{\beta_{ab}}=\frac{1}{\sqrt{2}}(\ket{0,b}+(-1)^a \ket{1,1 \oplus b}$
are the Bell states.  Here, we take $|c_1| \ge |c_2| \ge |c_3|$ with local unitary transformations. For this family of states, the eigenvalues satisfy the ordering $\lambda_{01} \ge  \lambda_{00} \ge \lambda_{10} \ge \lambda_{11}$. This family of states are entangled iff $\lambda_{01} > 1/2$. The geometric discord \cite{DVB10} of this family of states is given by  $D(\tau)=\frac{1}{2}(c^2_2+c^2_3)$.

\begin{proposition}
The Schr\"{o}dinger strength of the Bell-diagonal states in the two-setting steering scenario is given by
\ba
S_{SU_2}(\tau)=|c_2|. \label{SSBD2}
\ea 
\end{proposition}
\begin{proof}
Up to local unitary transformations, the Bell-diagonal states can be reduced to the family of states with  $c_1 \ge c_2 \ge |c_3|$, here $c_1,c_2 \ge 0$.  With this constraint,
the following decomposition for the Bell-diagonal states can be obtained:
\be \label{BDdec2SS}
\tau=c_2 \ketbra{\beta_{01}}{\beta_{01}} + (1-c_2) \rho_{sep},
\ee
where 
\begin{align}
\rho_{sep}&=\frac{1}{1-c_2}\bigg( 2 (\lambda_{00}-\lambda_{10}) \frac{ \ketbra{\beta_{00}}{\beta_{00}} +\ketbra{\beta_{01}}{\beta_{01}} }{2} \nonumber \\
&+2 \lambda_{10} \frac{ \ketbra{\beta_{00}}{\beta_{00}} +\ketbra{\beta_{10}}{\beta_{10}} }{2} + 2 \lambda_{11} \frac{ \ketbra{\beta_{01}}{\beta_{01}} +\ketbra{\beta_{11}}{\beta_{11}} }{2}\bigg).
\label{sepD=0}
\end{align}
It can be checked that, for the measurements $A_0=B_0=\sigma_x$ and $A_1=B_1=\sigma_y$, the correlation arising from the 
above Bell-diagonal states has the following canonical decomposition:
\begin{align} 
P(ab|xy;\tau) = &c_2 P_S^{Ext} \Big(ab|xy;\ket{\beta_{01}} \langle \beta_{01} | \Big) \nonumber\\
&+ (1-c_2) P_{US}(ab|xy;\rho_{sep}),
\label{optsst}
\end{align}
where $P_S^{Ext} \Big(ab|xy;\ket{\beta_{01}} \langle \beta_{01} | \Big)$  is an extremal steerable correlation arising from the state $\ket{\beta_{01}}$ as this correlation maximally violates steering inequality (\ref{FJWR}) with $n=2$ and has steering cost equal to $1$ \cite{DDJ+18}. The correlation  $P_{US}(ab|xy;\rho_{sep})$ is an unsteerable correlation, which is not superunsteerable (proof provided in Appendix \ref{ap1}). On the other hand, the box  $P(ab|xy;\tau)$ given by Eq. (\ref{optsst}) demonstrates quantumness in the form of either superunsteerability, or EPR steering for $c_2>0$ (proved in  Appendix \ref{ap2}). Hence, the decomposition (\ref{optsst}) corresponds to the canonical decomposition that implies the Schr\"{o}dinger strength of the box $P(ab|xy;\tau)$.

Note that in the above optimal decomposition of the Bell-diagonal states given by Eq. (\ref{BDdec2SS}) 
the weight of the entangled part is maximized over all 
possible decompositions such that the separable part cannot give rise to superunsteerability for
the measurements used for obtaining the box given by Eq. (\ref{optsst}). 
The Schr\"{o}dinger strength of the box $P(ab|xy;\tau)$ is equal 
to the weight of the entangled part in the decomposition given by Eq. (\ref{BDdec2SS}). This implies that 
there cannot be any other box arising from the Bell-diagonal state having more Schr\"{o}dinger strength. Therefore, 
the Schr\"{o}dinger strength of the box $P(ab|xy;\tau)$ is also the Schr\"{o}dinger strength of 
the Bell-diagonal state.
\end{proof}

\begin{proposition}
The Schr\"{o}dinger strength of the Bell-diagonal states in the three-setting steering scenario  is given by
\ba
S_{SU_3}(\tau)=|c_3|. \label{SSBD3}
\ea 
\end{proposition}
\begin{proof}
Up to local unitary transformations, the Bell-diagonal states can be reduced to the family of states with $c_1 \ge c_2 \ge |c_3|$, here $c_1,c_2 \ge 0$  and $c_3 \leq 0$. For this family of states, the eigenvalues satisfy the ordering
$\lambda_{01} \ge \lambda_{11} \ge  \lambda_{00} \ge \lambda_{10}$. With this constraint,
the following decomposition for the Bell-diagonal states can be obtained:
\be \label{stdebd3}
\tau=|c_3| \ketbra{\beta_{01}}{\beta_{01}} + (1-|c_3|) \rho_{sep},
\ee
where 
\begin{align}\label{sep3D=0}
\rho_{sep}&=\frac{1}{1-|c_3|}\bigg( 2 (\lambda_{00}-\lambda_{11}) \frac{ \ketbra{\beta_{00}}{\beta_{00}} +\ketbra{\beta_{01}}{\beta_{01}} }{2} \nonumber \\
&+2 \lambda_{10} \frac{ \ketbra{\beta_{01}}{\beta_{01}} +\ketbra{\beta_{10}}{\beta_{10}} }{2} + 2 \lambda_{11} \frac{ \ketbra{\beta_{00}}{\beta_{00}} +\ketbra{\beta_{11}}{\beta_{11}} }{2}\bigg).
\end{align} 
Note that, the above state $\rho_{sep}$ is a separable state. For the measurements $A_0=B_0=\sigma_x$, $A_1=B_1=\sigma_y$ and $A_2=B_2=\sigma_z$, the correlation arising from the 
above Bell-diagonal states has the following canonical decomposition:
\begin{align} 
P(ab|xy;\tau)= &|c_3| P_S^{Ext}\Big(ab|xy;\ket{\beta_{01}} \langle \beta_{01} | \Big)  \nonumber \\
&+ (1-|c_3|) P_{US}(ab|xy;\rho_{sep}),
\label{optsst3}
\end{align}
 where $P_S^{Ext}\Big(ab|xy;\ket{\beta_{01}} \langle \beta_{01} | \Big)$ is an extremal steerable correlation arising from the state $\ket{\beta_{01}}$ as this correlation maximally violates steering inequality (\ref{FJWR}) with $n=3$ and has steering cost equal to $1$ \cite{DDJ+18}.  The correlation $P_{US}(ab|xy;\rho_{sep})$ is an unsteerable correlation (in Appendix \ref{ap3}, we show that  $P_{US}(ab|xy;\rho_{sep})$ has a LHV-LHS model) which may have superunsteerability. 
 
Note that the extremal steerable box that maximally violates the linear steering inequality in the three setting scenario contains correlations in all three axes of the maximally entangled state.
 Superunsteerable states having correlations only in two axes never lead to a superunsteerable box which has a nonvanishing fraction of the extremal  steerable box. Correlations in all three axes of the superunsteerable two-qubit state are necessary for the state to generate the superunsteerable box having a nonzero fraction of an extremal steerable box. Since the separable state in the decomposition (\ref{stdebd3}) does not have correlations in all three axes, the superunsteerable box arising from this state does not have a nonvanishing
  fraction of the extremal steerable box.
 Hence, the fraction of the extremal steerable box in
the decomposition (\ref{optsst3}) is the maximal steerable box fraction,  implying an optimal decomposition with the Schr\"{o}dinger strength of the box $P(ab|xy;\tau)$.

 Thus, the above Bell diagonal state demonstrates quantumness either in the form of steering, or in the form of superunsteerability with a nonzero Schr\"{o}dinger strength 
  in the three-settings scenario if and only if $|c_3|$ is non-zero. We can, therefore, conclude that the Schr\"{o}dinger strength of the  Bell diagonal state is maximized by the above measurements, and 
the Schr\"{o}dinger strength of the box (\ref{optsst3}) also gives the Schr\"{o}dinger strength of the Bell-diagonal states.
\end{proof}

\section{Schr\"{o}dinger strength optimizes Random Access Codes}\label{sec4} 
 
  Let us consider the Bell-diagonal states with $|c_1| \ge |c_2| \ge |c_3|$.
 In the case of quantum $2 \xrightarrow{p}1$ RAC, the encoding operation  
 using  correlations along the $x$ and $y$ axes provides  efficiency  given by \cite{BP14}
 \be
   {{P}_{{\rm min} }}(\tau)=\frac{1}{2}\left(1+ \frac{1}{\sqrt{c^{-2}_1+c^{-2}_2}} \right). \label{2to1Eff}
   \ee

 Next, considering the case of the quantum $3 \xrightarrow{p}1$ RAC, the encoding operation  
 using  correlations along all three axes of the Bell-diagonal  provides  efficiency  given by \cite{BP14}
  is given by \cite{BP14}
 \be
   {{P}_{{\rm min} }}(\tau)=\frac{1}{2}\left(1+ \frac{1}{\sqrt{c^{-2}_1+c^{-2}_2+c^{-2}_3}} \right). \label{3to1Eff}
 \ee  
 
Now we will present two propositions which will enable us to present the main result of this paper.

\begin{proposition} 
Within the separable Bell-diagonal states with $c_1,c_2 \ge 0$, $c_3 \le 0$ and $c_1 \ge c_2 \ge |c_3|$,  
the state which has the maximal Schr\"{o}dinger strength $S_{SU_2}=1/2$ in the two-setting scenario 
corresponds to the state with $c_1=c_2=\frac{1}{2}$ and $c_3=0$.
\end{proposition}
\begin{proof}
Note that the Bell-diagonal states  with 
 $ c_1,c_2 \ge 0, c_3 \le 0$ and $c_1 \ge c_2 \ge |c_3|$ are 
separable iff the largest eigenvalue  $\lambda_{01} \le 1/2$. Therefore,
 $ c_1 +c_2 - c_3  \le 1$. From the positivity of the lowest eigenvalue of the Bell-diagonal states, it follows that 
 $ c_1 +c_2 + c_3  \le 1$. Combining the above equations,  we get
 $c_1 +c_2 \le 1$ which for our chosen restrictions on the coefficients defined above, leads to $ 2c_2  \le  c_1 + c_2  \le  1$. This
 implies that $c_2 \le 1/2$ within separable Bell-diagonal states.
 Therefore, the separable Bell-diagonal state which has $c_1=c_2=1/2$ and $c_3=0$
 has the maximal Schr\"{o}dinger strength  in the context of the two-settings scenario. 
\end{proof}  
 The above values of the correlation coefficients $c_1=c_2=\frac{1}{2}$ and $c_3=0$, lead to the   
 efficiency ${{P}_{{\rm min} }}=\frac{1}{2}(1+\frac{1}{2\sqrt{2}})\approx 0.677$,  higher than the classical bound ${{P}^{2\to 1}_{{\rm cl} }}=0.667$.  This is obtained by using Eq.(\ref{2to1Eff}).

\begin{proposition}
Within the separable Bell-diagonal states with $c_1,c_2 \ge 0$, $c_3 \le 0$ and $c_1 \ge c_2 \ge |c_3|$,  
the state which has the maximal  Schr\"{o}dinger strength $S_{SU_3}=1/3$ in the three-setting scenario corresponds to the state with $|c_1|=|c_2|=|c_3|=\frac{1}{3}$.
\end{proposition}
\begin{proof}
Note that the correlation coefficients of the separable Bell-diagonal states with the constraints $c_1,c_2 \ge 0$, $c_3 \le 0$ and $c_1 \ge c_2 \ge |c_3|$ satisfy the relations
$ c_1 + c_2 + |c_3| \le 1$, and 
 $ 3|c_3|  \le  c_1 + c_2 + |c_3|   \le  1$.
  This implies that $|c_3| \le 1/3$ within separable Bell-diagonal states.
  Therefore, the  separable Bell-diagonal state  which has $c_1=c_2=-c_3=1/3$
 has the maximal Schr\"{o}dinger strength  in the context of the three-settings scenario.
\end{proof} 
 The above values of the correlation coefficients $|c_1|=|c_2|=|c_3|=\frac{1}{3}$, lead to the   
 efficiency ${{P}_{{\rm min} }}=\frac{1}{2}(1+\frac{1}{3\sqrt{3}})\approx 0.596$, considerably above the classical bound ${{P}^{3\to 1}_{{\rm cl} }}=0.5$. This is calculated by using Eq.(\ref{3to1Eff}).

 It may be noted that the separable Bell-diagonal states with $c_1=c_2=1/2$ 
 and  
 $c_1=c_2=-c_3=1/3$  
 are the optimal separable states for 
 the $2\to 1$ and $3\to 1$ quantum codes, respectively \cite{BP14}. Since these states have maximal Schr\"{o}dinger strengths 
 in the two- and three-setting scenarios, respectively, within the separable Bell-diagonal states, 
 their corresponding values identify the optimal separable Bell-diagonal states for the $2\to 1$ and $3\to 1$ RAC protocols, respectively. On the other hand, the efficiencies (\ref{3to1Eff}) and (\ref{2to1Eff})
 are not monotonic functions of geometric  discord for different classes of states \cite{BP14}. Hence, 
the amount of geometric discord
does not specify the state yielding optimal efficiency of the protocol. Therefore, geometric discord does not serve as the resources for the above two RACs. Here, we identify superunsteerability as the unique resource for RACs assisted by Bell-diagonal states. 
   
   Note that the correlation coefficients in the Bell-diagonal states can be experimentally 
determined by measuring the Pauli observables since $c_i=\braket{\sigma_i \otimes \sigma_i}=\tr(\tau \sigma_i \otimes \sigma_i)$. From these joint expectation values,  the Schr\"{o}dinger strength of the Bell-diagonal states in the two- and three-setting scenarios can be determined  as the second and third largest of these three  joint expectation values, respectively.

\section{Conclusions}\label{sec5} 

In this work, we propose a hitherto unexplored resource for quantum information processing by unsteerable states.  We formulate a quantifier of nonclassical correlation going beyond steering which we call Schr\"{o}dinger strength. A non-vanishing value of this measure 
 demonstrates a phenomenon called superunsteerability \cite{DBD+17} which can occur even
for separable states.  We derive analytical expressions 
for this measure for Bell-diagonal states in the two- and three-setting steering scenarios which are related to the $2  \to 1$ and $3 \to 1$ RAC protocols, respectively, assisted
by finite shared randomness limited to two bits  within separable Bell diagonal states. We demonstrate that the Schr\"{o}dinger strength identifies the optimal states within separable Bell-diagonal states for both these protocols.  In other words, the maximal value of Schr\"{o}dinger strength for the
 corresponding separable Bell-diagonal state provides the optimal quantum efficiency of $2  \to 1$ and $3 \to 1$ RAC protocols assisted by separable Bell-diagonal states.

 The analysis presented here provides an unambiguous answer to the question as to which is the precise quantum resource for RACs  using separable Bell-diagonal states in the presence of finite shared randomness \cite{BP14}.  The fact that quantumness defined by superunsteerability is a more appropriate 
resource compared to discord for enabling RACs may be understood as follows. While discord is a property exclusive to
the given state, superunsteerabilty further captures the operational characteristics of the 
protocol, as we have shown through the linkage of steering scenarios with the corresponding RACs.  It remains to be studied  whether superunsteerability is the unique resource for the above two RACs in the presence of finite shared randomness using other separable states.
Moreover, further studies of local resources such as superlocality \cite{DW15,JAS17,JDS+18}
and superunsteerability \cite{DBD+17,DAS201855} for other quantum information processing tasks using
separable states admitting local correlations, should be worthwhile.  Such studies may provide the scope for shedding further light on the larger question of the boundary between classical and quantum theory.

\section*{Acknowledgements} C.J. acknowledges S. N. Bose Centre,
Kolkata for the postdoctoral fellowship and partial financial support
from the Ministry of Science and Technology, Taiwan (Grant
No.  108-2811-M-006-501). 
D.D. acknowledges the financial support from the University Grants
Commission (UGC), Government of India.
S.K. acknowledges
the support of a fellowship from Bose Institute, Kolkata.

\appendix

\section{The box $P_{US}(ab|xy;\rho_{sep})$ given by Eq. (\ref{optsst}) is unsteerable and  not super-unsteerabie} \label{ap1}

The correlation $P_{US}(ab|xy;\rho_{sep})$  given in Eq. (\ref{optsst}) (with $a,b,x,y \in \{0,1\}$) has the following LHV-LHS model with dimension of the hidden variable being equal to $2$:
\begin{equation}
P_{US}(ab|xy;\rho_{sep}) = \sum^{1}_{\lambda=0} p(\lambda) P(a|x, \lambda) P(b|y; \rho_{\lambda})
\end{equation}
where, $p(0) = p(1) = \frac{1}{2}$. 
\begin{equation}      
P(a|x, 0) =\begin{tabular}{c|cc}
\backslashbox{(y)}{(b)} & (0) & (1) \\\hline
(0) & $1$ & $0$  \\
(1) & $\frac{1}{2}$ & $\frac{1}{2}$  \\
\end{tabular},
\end{equation}
where each row and column corresponds to a fixed measurement $(y)$ and a fixed outcome $(b)$ respectively. Throughout the paper we will follow the same convention.
\begin{equation}      
P(a|x, 1) =\begin{tabular}{c|cc}
\backslashbox{(y)}{(b)} & (0) & (1) \\\hline
(0) & $0$ & $1$  \\
(1) & $\frac{1}{2}$ & $\frac{1}{2}$  \\
\end{tabular}.
\end{equation}
On the other hand, we have,
\begin{equation}      
P(b|y; \rho_{0}) =\begin{tabular}{c|cc}
\backslashbox{(y)}{(b)} & (0) & (1) \\\hline
(0) & $\dfrac{1+c_1-2 c_2}{2(1-c_2)}$ & $\dfrac{1- c_1}{2(1-c_2)}$  \\
(1) & $\frac{1}{2}$ & $\frac{1}{2}$  \\
\end{tabular}.
\end{equation}
Note that for $1 \ge c_1 \ge c_2 \ge |c_3|$ and $c_1,c_2 \ge 0$ each of the probability distributions satisfies $0 \leq P(b|y; \rho_{0}) \leq 1$. The above correlation on Bob's side can be reproduced by performing projective measurements corresponding to the observables $B_0=\sigma_x$ and $B_1=\sigma_y$ on the qubit state given by,
\begin{equation}
|\psi_0\rangle = \sqrt{\dfrac{1+c_1-2 c_2}{2(1-c_2)}} | + \rangle + \sqrt{\dfrac{1- c_1}{2(1-c_2)}} | - \rangle,
\end{equation}
where $|+\rangle$ and $|-\rangle$ are the eigenstates of the operator $\sigma_x$ corresponding to the eigenvalues $+1$ and $-1$ respectively. 
and
\begin{equation}      
P(b|y; \rho_{1}) =\begin{tabular}{c|cc}
\backslashbox{(y)}{(b)} & (0) & (1) \\\hline
(0) & $\dfrac{1- c_1}{2(1-c_2)}$ & $\dfrac{1 + c_1 - 2 c_2}{2(1-c_2)}$  \\
(1) & $\frac{1}{2}$ & $\frac{1}{2}$  \\
\end{tabular}.
\end{equation}
Note that for $1 \ge c_1 \ge c_2 \ge |c_3|$ and $c_1,c_2 \ge 0$ each of the probability distributions satisfies $0 \leq P(b|y, \rho_{1}) \leq 1$. The above correlation on Bob's side can be reproduced by performing projective measurements corresponding to the observables $B_0=\sigma_x$ and $B_1=\sigma_y$ on the qubit state given by,
\begin{equation}
|\psi_1\rangle = \sqrt{\dfrac{1- c_1}{2(1-c_2)}}  | + \rangle + \sqrt{\dfrac{1+c_1-2 c_2}{2(1-c_2)}} | - \rangle.
\end{equation}

Hence, the box $P_{US}(ab|xy;\rho_{sep})$ has LHV-LHS model with hidden variable dimension $2$. On the other hand, this box is produced by performing appropriate measurements on the $2 \otimes 2$ state given by Eq.(\ref{sepD=0}). Hence, the box $P_{US}(ab|xy;\rho_{sep})$ is unsteerable and is not super-unsteerable.

\section{Quantumness in the form of superunsteerability or EPR steering of the box $P(ab|xy;\tau)$ given by Eq. (\ref{optsst}) for $c_2 >0$.} \label{ap2}

Before going into the details, let us first define some notations which we will use later. In the case of two inputs per party and two-outputs per input, there are four possible deterministic boxes $P^{\alpha \beta}_D(a|x)$ (or, simply $P^{\alpha \beta}_D$)  on Alice's end and four possible deterministic boxes $P_D^{\gamma\epsilon}(b|y)$ (or, simply $P_D^{\gamma\epsilon}$) on Bob's end. These are given by,
\begin{equation}
P_D^{\alpha\beta}(a|x)=\left\{
\begin{array}{lr}
1, & a=\alpha x\oplus \beta\\
0 , & \text{otherwise}\\
\end{array}
\right. 
\label{}
\end{equation} 
and 
\begin{equation}
P_D^{\gamma\epsilon}(b|y)=\left\{
\begin{array}{lr}
1, & b=\gamma x\oplus \epsilon\\
0 , & \text{otherwise}.\\
\end{array}
\right.   
\label{}
\end{equation} 
Here, $\alpha,\beta,\gamma,\epsilon\in  \{0,1\}$ and  $\oplus$ denotes
addition modulo  $2$.

The box $P(ab|xy;\tau)$ given in Eq. (\ref{optsst}) (with $a,b,x,y \in \{0,1\}$) has the following form:
 \begin{equation}      
 \label{matrixptau}
P(ab|xy;\tau) =\begin{tabular}{c|cccc}
 \backslashbox{(x,y)}{(a,b)} & (0,0) & (0,1) & (1,0) & (1,1)\\\hline
(0,0) & $\frac{1+c_1}{4}$ & $\frac{1-c_1}{4}$ & $\frac{1-c_1}{4}$ & $\frac{1+c_1}{4}$ \\
(0,1) & $\frac{1}{4}$ & $\frac{1}{4}$ & $\frac{1}{4}$ & $\frac{1}{4}$ \\
(1,0) & $\frac{1}{4}$ & $\frac{1}{4}$ & $\frac{1}{4}$ & $\frac{1}{4}$ \\
(1,1) & $\frac{1+c_2}{4}$ & $\frac{1-c_2}{4}$ & $\frac{1-c_2}{4}$ & $\frac{1+c_2}{4}$ \\
\end{tabular},
\end{equation}
where each row and column corresponds to a fixed measurement setting $(xy)$ and a fixed outcome $(ab)$ respectively. $1 \ge c_1 \ge c_2 \ge |c_3|$ and $c_1,c_2 \ge 0$.

Let us try to construct LHV-LHS model of the box $P(ab|xy;\tau)$ with the dimension of the hidden variable being equal to $2$. Before proceeding, we want  to mention that in  case of the box $P(ab|xy;\tau)$, all the
marginal  probability  distributions of  Alice and  Bob  are maximally mixed:
\begin{equation}
\label{mar}
p_{\text{Alice}}(a|x;\tau)  =  p_{\text{Bob}}(b|y;\tau)  =  \frac{1}{2}  \forall a,b,x,y.
\end{equation}  
Let us now try  to  check whether  the box $P(ab|xy;\tau)$
can be decomposed in the following form:
\begin{equation} \label{model2dim}
P(ab|xy;\tau) = \sum^{1}_{\lambda=0} p(\lambda) P(a|x, \lambda) P(b|y; \rho_{\lambda}),
\end{equation}
where $p(0)=x_0$, $p(1)=x_1$  ($0 <x_0<1$, $0  <x_1<1$, $x_0+x_1  =1$). Let us  assume that Alice's
strategy   to be  deterministic   one, i. e.,  each   of  the   two  boxes $P(a|x, 0)$ and  $P(a|x, 1)$ in the above decomposition
belongs to any one among $P_D^{00}$, $P_D^{01}$, $P_D^{10}$ and $P_D^{11}$. In
order to satisfy the  marginal probabilities for Alice $p_{\text{Alice}}(a|x;\tau) = \frac{1}{2}$ $\forall a,x$,
the only  two possible choices  of $P(a|x, 0)$  and $P(a|x, 1)$ are:
\begin{enumerate}
\item $P_D^{00}$ and $P_D^{01}$ with $x_0=x_1=\frac{1}{2}$
\item $P_D^{10}$ and $P_D^{11}$ with $x_0=x_1=\frac{1}{2}$.
\end{enumerate}

Now, it can be easily checked  that none of these two possible choices
will  satisfy  all  the  bipartite  joint  probability  distributions of the box $P(ab|xy;\tau)$ with $c_2 >0$,  simultaneously, for any choice of $P(b|y; \rho_{0})$ and $P(b|y; \rho_{1})$. 
It  is, therefore,  impossible to  construct LHV-LHS model of the box $P(ab|xy;\tau)$ with the dimension of the hidden variable being equal to $2$  where Alice uses different deterministic strategy at each $\lambda$.

The box $P(ab|xy;\tau)$ cannot be reproduced
by a LHV-LHS model with hidden  variable of dimension $2$
even if Alice uses \textit{nondeterministic strategy} for each $\lambda$.
To see this, we note that in $2 - 2 - 2$ Bell-scenario (involving $2$ parties, $2$ measurement settings per party, $2$ outcomes per setting), hidden
variable with dimension $d_\lambda \leq 4$ is sufficient for reproducing
any local correlation \cite{DW15}. Since unsteerable correlations form
a subset of the local correlations, in $2 - 2 - 2$ steering-scenario
hidden variable with dimension $d_\lambda \leq 4$ is sufficient for
reproducing any unsteerable correlation. Hence, a LHV-LHS model with hidden  variable of dimension $2$ of the  box $P(ab|xy;\tau)$ can be achieved by constructing a LHV-LHS model of the  box  $P(ab|xy;\tau)$ with hidden variable of dimension $3$ or $4$ with different deterministic distributions at Alice's side 
followed by taking equal probability distributions at Bob's side as common and making the 
corresponding distributions at Alice's side non-deterministic.

Let us now try check whether the  box $P(ab|xy;\tau)$ can be simulated by  a LHV-LHS model with hidden  variable of dimension $3$ where Alice uses different deterministic strategy at each $\lambda$. 
In this case, we assume that the box can be decomposed in the following way:
\begin{equation}
P(ab|xy;\tau) = \sum^{2}_{\lambda=0} p(\lambda) P(a|x, \lambda) P(b|y; \rho_{\lambda}),
\end{equation}
Here, $p(0)= x_0$,  $p(1) = x_1$, $p(2)  = x_2$ ($0 <x_0<1$, $0  <x_1<1$, $0 <x_2<1$,
$x_0+x_1+x_2 =1$) and $P(a|x, \lambda)$ are different deterministic boxes. 
Since Alice's boxes are  deterministic, the three boxes 
$P(a|x, 0)$, $P(a|x, 1)$ and  $p(a|x, 2)$  must be  equal to  any three among $P_D^{00}$,
$P_D^{01}$, $P_D^{10}$  and $P_D^{11}$. But any  such combination will
not satisfy  the marginal  probabilities $p_{\text{Alice}}(a|x;\tau) = \frac{1}{2}$ $\forall a,x$ for  Alice. This implies that the  box $P(ab|xy;\tau)$ cannot be simulated by  a LHV-LHS model with hidden  variable of dimension $3$ where Alice uses different deterministic strategy at each $\lambda$. 

Therefore, we can state that if the box $P(ab|xy;\tau)$ with $c_2 >0$ has LHV-LHS model, then the dimension of the hidden variable in that LHV-LHS model must be equal to $4$ when Alice uses deterministic strategies.  

Suppose  the box $P(ab|xy;\tau)$ has the following LHV-LHS model:
\begin{equation}
\label{new1}
P(ab|xy;\tau) = \sum_{\lambda=0}^{3} p(\lambda) P(a|x, \lambda) P(b|y; \rho_{\lambda}),
\end{equation} 
where $P(a|x, \lambda)$ are different deterministic boxes
and either any three of the four boxes $P(b|y; \rho_{\lambda})$ 
are equal to each other, or there exists two sets each containing two equal boxes $P(b|y; \rho_{\lambda})$; $0 < p(\lambda) < 1$ for $\lambda$ = $0,1,2,3$; 
$\sum_{\lambda=0}^{3} p(\lambda) = 1$. Then taking equal boxes 
$P(b|y; \rho_{\lambda})$ at Bob's side as common and making corresponding box 
at Alice's side non-deterministic will reduce the dimension of the hidden variable
from $4$ to $2$. For example, let us consider
\be
P(b|y; \rho_0) = P(b|y, \rho_1) = P(b|y; \rho_2).
\ee
Now in order to satisfy Alice's marginal given by Eq. (\ref{mar}), one must take $p(0)$ = $p(1)$ = $p(2)$ = $p(3)$ = $\frac{1}{4}$. 
Hence, the decomposition (\ref{new1}) can be written as, 
\begin{align}
P(ab|xy;\tau) = & q(0) Q(a|x, 0) P(b|y; \rho_0) \nonumber \\
&+ p(3) P(a|x, 3) P(b|y; \rho_3),
\label{new2}
\end{align}
where  
\be
Q(a|x, 0) = \dfrac{P(a|x, 0) + P(a|x, 1) + P(a|x, 2)}{3},
\ee
which is a non-deterministic box at Alice's side, and
$q(0) = \frac{3}{4}$
The decomposition (\ref{new2}) represents a LHV-LHS model of the box $P(ab|xy;\tau)$ with hidden  variable of dimension $2$ with different deterministic/non-deterministic boxes at Alice's side. Now in this protocol, considering arbitrary boxes $P(b|y; \rho_{\lambda})$ at Bob's side (without considering any constraint),
it can be checked
that all the bipartite distributions of the box $P(ab|xy;\tau)$ with $c_2 >0$ are not reproduced simultaneously. Hence, this will hold when $P(b|y; \rho_{\lambda})$ have quantum realizations.

There are the following other cases in which the dimension of the hidden variable in the LHV-LHS model of the box $P(ab|xy;\tau)$
can be reduced from $4$ to $2$: 
\begin{center}
$P(b|y; \rho_0) = P(b|y; \rho_2) = P(b|y; \rho_3)$;  \\ 
$P(b|y; \rho_0) = P(b|y; \rho_1) = P(b|y; \rho_3)$;   \\
$P(b|y; \rho_1) = P(b|y; \rho_2) = P(b|y; \rho_3)$;  \\
$P(b|y; \rho_0) = P(b|y; \rho_1)$ as well as $P(b|y; \rho_2) = P(b|y; \rho_3)$;\\ 
$P(b|y; \rho_0) = P(b|y; \rho_2)$ as well as $P(b|y; \rho_1) = P(b|y; \rho_3)$;\\ 
$P(b|y; \rho_0) = P(b|y; \rho_3)$ as well as $P(b|y; \rho_1) = P(b|y; \rho_2)$. \\
\end{center}
Now in any of these possible cases, considering arbitrary boxes 
$P(b|y; \rho_{\lambda})$ at Bob's side (without considering any constraint), it can be checked that all the bipartite distributions of the box $P(ab|xy;\tau)$ with $c_2 >0$ are not reproduced simultaneously. 
Hence, this also holds when the boxes $P(b|y; \rho_{\lambda})$ have quantum realizations.

Hence, either the box $P(ab|xy;\tau)$ with $c_2 >0$ does not have LHV-LHS model or it is impossible to reduce the dimension of the hidden variable in the LHV-LHS model of the box $P(ab|xy;\tau)$ with $c_2 >0$ from $4$ to $2$.

It  is, therefore,  impossible to  reproduce the box  $P(ab|xy;\tau)$ (with $c_2 >0$) by LHV-LHS model with hidden variable dimension $2$ and with deterministic/non-deterministic 
boxes at Alice's side.

It can be checked that the box $P(ab|xy;\tau)$ with $c_2 >0$  is  non-product for $c_2 > 0$.  It  is,  therefore,
impossible to construct a LHV-LHS model of the box  $P(ab|xy;\tau)$ (with $c_2 >0$) with hidden variable of dimension $1$. 

Hence, one can conclude that the box $P(ab|xy;\tau)$ with $c_2 > 0$ is either steerable or has LHV-LHS model with the dimension of the hidden variable being greater than $2$. On the other hand, the box is produced from the $2 \otimes 2$ Bell diagonal state given by Eq.(\ref{BD}). 
Therefore, the box $P(ab|xy;\tau)$ with $c_2 > 0$ is either steerable or super-unsteerable.

\section{The box $P_{US}(ab|xy;\rho_{sep})$ given by Eq. (\ref{optsst3}) is unsteerable} \label{ap3}

The correlation $P_{US}(ab|xy;\rho_{sep})$ given in Eq. (\ref{optsst3}) (with $a,b \in \{0,1\}$ and with $x,y \in \{0,1, 2\}$) has the following LHV-LHS model with dimension of the hidden variable being equal to $4$:
\begin{equation}
P_{US}(ab|xy;\rho_{sep}) = \sum^{3}_{\lambda=0} p(\lambda) P(a|x, \lambda) P(b|y; \rho_{\lambda})
\end{equation}
where, $p(0) = p(1) = p(2) = p(3) = \frac{1}{4}$. 
\begin{equation}      
P(a|x, 0) =\begin{tabular}{c|cc}
\backslashbox{(y)}{(b)} & (0) & (1) \\\hline
(0) & $1$ & $0$  \\
(1) & $1$ & $0$  \\
(2) & $\frac{1}{2}$ & $\frac{1}{2}$  \\
\end{tabular},
\end{equation}

\begin{equation}      
P(a|x, 1) =\begin{tabular}{c|cc}
\backslashbox{(y)}{(b)} & (0) & (1) \\\hline
(0) & $1$ & $0$  \\
(1) & $0$ & $1$  \\
(2) & $\frac{1}{2}$ & $\frac{1}{2}$  \\
\end{tabular},
\end{equation}
\begin{equation}      
P(a|x, 2) =\begin{tabular}{c|cc}
\backslashbox{(y)}{(b)} & (0) & (1) \\\hline
(0) & $0$ & $1$  \\
(1) & $1$ & $0$  \\
(2) & $\frac{1}{2}$ & $\frac{1}{2}$  \\
\end{tabular},
\end{equation}
and 
\begin{equation}      
P(a|x, 3) =\begin{tabular}{c|cc}
\backslashbox{(y)}{(b)} & (0) & (1) \\\hline
(0) & $0$ & $1$  \\
(1) & $0$ & $1$  \\
(2) & $\frac{1}{2}$ & $\frac{1}{2}$  \\
\end{tabular}.
\end{equation}
On the other hand, we have,
\begin{equation}      
P(b|y; \rho_{0}) =\begin{tabular}{c|cc}
\backslashbox{(y)}{(b)} & (0) & (1) \\\hline
(0) & $\dfrac{1+c_1+2 c_3}{2(1+c_3)}$ & $\dfrac{1-c_1}{2(1+c_3)}$  \\
(1) & $\dfrac{1+c_2+2 c_3}{2(1+c_3)}$ & $\dfrac{1-c_2}{2(1+c_3)}$  \\
(2) & $\frac{1}{2} [ 1 - f(c_1, c_2, c_3)]$ & $\frac{1}{2} [ 1 + f(c_1, c_2, c_3)]$  \\
\end{tabular},
\end{equation}
where, $f(c_1, c_2, c_3)$ = $\frac{\sqrt{(1-c_1)(1+c_1+ 2 c_3)-(c_2 + c_3)^2}}{1+c_3}$. Note that for $1 \ge c_1 \ge c_2 \ge |c_3|$; $c_1,c_2 \ge 0$ and $c_3 \leq 0$ each of the probability distributions satisfies $0 \leq P(b|y, \rho_{0}) \leq 1$. The above correlation on Bob's side can be reproduced by performing projective measurements corresponding to the observables $B_0=\sigma_x$, $B_1=\sigma_y$ and $B_2 = \sigma_z$ on the qubit state given by,
\begin{equation}
|\psi_0\rangle = \sqrt{\frac{1+c_1+2 c_3}{2(1+c_3)}} | + \rangle + e ^{i \phi_0} \sqrt{\frac{1-c_1}{2(1+c_3)}} | - \rangle,
\end{equation}
where $|+\rangle$ and $|-\rangle$ are the eigenstates of the operator $\sigma_x$ corresponding to the eigenvalues $+1$ and $-1$ respectively. $ \phi_0 =\sin^{-1}  \frac{c_2 + c_3}{\sqrt{1-c_1^2 - 2 c_1 c_3 + 2 c_3}}$. It can be easily checked that $|\frac{c_2 + c_3}{\sqrt{1-c_1^2 - 2 c_1 c_3 + 2 c_3}}|$ $\leq 1$ for $1 \ge c_1 \ge c_2 \ge |c_3|$; $c_1,c_2 \ge 0$; $c_3 \leq 0$ and $c_1 + c_2 + c_3 \leq 0$ (which comes from the positivity condition of the lowest eigenvalue of the Bell diagonal state).

\begin{equation}      
P(b|y; \rho_{1}) =\begin{tabular}{c|cc}
\backslashbox{(y)}{(b)} & (0) & (1) \\\hline
(0) & $\dfrac{1+c_1+2 c_3}{2(1+c_3)}$ & $\dfrac{1-c_1}{2(1+c_3)}$  \\
(1) & $\dfrac{1-c_2}{2(1+c_3)}$ & $\dfrac{1+c_2+2 c_3}{2(1+c_3)}$  \\
(2) & $\frac{1}{2} [ 1 + f(c_1, c_2, c_3)]$ & $\frac{1}{2} [ 1 - f(c_1, c_2, c_3)]$  \\
\end{tabular},
\end{equation}
where, $f(c_1, c_2, c_3)$ = $\frac{\sqrt{(1-c_1)(1+c_1+ 2 c_3)-(c_2 + c_3)^2}}{1+c_3}$. Note that for $1 \ge c_1 \ge c_2 \ge |c_3|$; $c_1,c_2 \ge 0$ and $c_3 \leq 0$ each of the probability distributions satisfies $0 \leq P(b|y, \rho_{1}) \leq 1$. The above correlation on Bob's side can be reproduced by performing projective measurements corresponding to the observables $B_0=\sigma_x$, $B_1=\sigma_y$ and $B_2 = \sigma_z$ on the qubit state given by,
\begin{equation}
|\psi_1\rangle = \sqrt{\frac{1+c_1+2 c_3}{2(1+c_3)}} | + \rangle + e ^{i \phi_1} \sqrt{\frac{1-c_1}{2(1+c_3)}} | - \rangle,
\end{equation}
where $ \phi_1 = \pi + \sin^{-1}  \frac{c_2 + c_3}{\sqrt{1-c_1^2 - 2 c_1 c_3 + 2 c_3}}$.

\begin{equation}      
P(b|y; \rho_{2}) =\begin{tabular}{c|cc}
\backslashbox{(y)}{(b)} & (0) & (1) \\\hline
(0) & $\dfrac{1-c_1}{2(1+c_3)}$ & $\dfrac{1+c_1+2 c_3}{2(1+c_3)}$ \\
(1) & $\dfrac{1+c_2+2 c_3}{2(1+c_3)}$ & $\dfrac{1-c_2}{2(1+c_3)}$  \\
(2) & $\frac{1}{2} [ 1 + f(c_1, c_2, c_3)]$ & $\frac{1}{2} [ 1 - f(c_1, c_2, c_3)]$  \\
\end{tabular},
\end{equation}
where, $f(c_1, c_2, c_3)$ = $\frac{\sqrt{(1-c_1)(1+c_1+ 2 c_3)-(c_2 + c_3)^2}}{1+c_3}$. Note that for $1 \ge c_1 \ge c_2 \ge |c_3|$; $c_1,c_2 \ge 0$ and $c_3 \leq 0$ each of the probability distributions satisfies $0 \leq P(b|y, \rho_{2}) \leq 1$. The above correlation on Bob's side can be reproduced by performing projective measurements corresponding to the observables $B_0=\sigma_x$, $B_1=\sigma_y$ and $B_2 = \sigma_z$ on the qubit state given by,
\begin{equation}
|\psi_2\rangle = \sqrt{\frac{1-c_1}{2(1+c_3)}} | + \rangle + e ^{i \phi_2} \sqrt{\frac{1+c_1+2 c_3}{2(1+c_3)}} | - \rangle,
\end{equation}
where $ \phi_2 = \pi - \sin^{-1}  \frac{c_2 + c_3}{\sqrt{1-c_1^2 - 2 c_1 c_3 + 2 c_3}}$.

\begin{equation}      
P(b|y; \rho_{3}) =\begin{tabular}{c|cc}
\backslashbox{(y)}{(b)} & (0) & (1) \\\hline
(0) & $\dfrac{1-c_1}{2(1+c_3)}$ & $\dfrac{1+c_1+2 c_3}{2(1+c_3)}$ \\
(1) & $\dfrac{1-c_2}{2(1+c_3)}$ & $\dfrac{1+c_2+2 c_3}{2(1+c_3)}$  \\
(2) & $\frac{1}{2} [ 1 - f(c_1, c_2, c_3)]$ & $\frac{1}{2} [ 1 + f(c_1, c_2, c_3)]$  \\
\end{tabular},
\end{equation}
where, $f(c_1, c_2, c_3)$ = $\frac{\sqrt{(1-c_1)(1+c_1+ 2 c_3)-(c_2 + c_3)^2}}{1+c_3}$. Note that for $1 \ge c_1 \ge c_2 \ge |c_3|$; $c_1,c_2 \ge 0$ and $c_3 \leq 0$ each of the probability distributions satisfies $0 \leq P(b|y, \rho_{3}) \leq 1$. The above correlation on Bob's side can be reproduced by performing projective measurements corresponding to the observables $B_0=\sigma_x$, $B_1=\sigma_y$ and $B_2 = \sigma_z$ on the qubit state given by,
\begin{equation}
|\psi_3\rangle = \sqrt{\frac{1-c_1}{2(1+c_3)}} | + \rangle + e ^{i \phi_3} \sqrt{\frac{1+c_1+2 c_3}{2(1+c_3)}} | - \rangle,
\end{equation}
where $ \phi_3 = - \sin^{-1}  \frac{c_2 + c_3}{\sqrt{1-c_1^2 - 2 c_1 c_3 + 2 c_3}}$.

\bibliography{JM}

\end{document}